\documentclass[12pt]{article}
\usepackage{amsmath,amsthm,amssymb,mathptmx,fullpage}
\newtheorem{theorem}{Theorem}
\newtheorem{lemma}{Lemma}
\begin{document}

\title{Not Every Domain of a Plain Decompressor Contains the Domain
of a Prefix-Free One\thanks{Partially supported by RFBR 0901-00709-a and by NAFIT ANR-08-EMER-008-01 grants}}

\author{Mikhail Andreev\thanks{Moscow State Lomonosov University, Mathematics Department, Logic and Algorithms Theory Division, Moscow, Russia, \texttt{amishaa@mail.ru}}
 \and
	   Ilya Razenshteyn\thanks{Moscow State Lomonosov University, Mathematics Department, Logic and Algorithms Theory Division, Moscow, Russia, \texttt{ilyaraz@gmail.com}} \and
	      Alexander Shen\thanks{LIF Marseille, CNRS \& University Aix--Marseille, France, on leave from Institute of Information Transmission Problems, Moscow, Russia, \texttt{alexander.shen@lif.univ-mrs.fr}, \texttt{sasha.shen@gmail.com}}
	      }

\maketitle

\begin{abstract}
C.~Calude, A.~Nies, L.~Staiger, and F.~Stephan  posed the following question about the relation between plain and prefix Kolmogorov complexities (see their paper in DLT 2008 conference proceedings): does the domain of every optimal decompressor contain the domain of some optimal prefix-free decompressor? In this paper we provide a negative answer to this question.
\end{abstract}

\section{Introduction}
\label{sec:introduction}

Let $D \colon \{0,1\}^* \to \{0,1\}^*$ be a computable partial function (used as a decompressor). Kolmogorov complexity of $x \in \{0,1\}^*$ with respect to $D$ is defined as the length of its shortest $D$-description:
        $$
C_D(x) := \min_{\{y\mid D(y) = x\}} l(y).
        $$
There exists an optimal decompressor $U$ such that $C_{U}$ is minimal up to $O(1)$. $C_{U}(x)$ is called \emph{plain complexity} of $x$ and is usually denoted by $C(x)$.

A decompressor is called \emph{prefix-free} if its domain is prefix-free (if $u$ is a prefix of $v$, the decompressor cannot be defined on both $u$ and $v$). Again it can be proved that there exists an optimal (up to $O(1)$) prefix-free decompressor $V$. $C_{V}(x)$ is called \emph{prefix complexity} of $x$ and is usually denoted by $K(x)$. (See, e.g., \cite{uppsala-notes} for more details.)

In~\cite{cnss} Calude et al. characterized domains of optimal plain and prefix decompressors.
They did not show any relation between domains of optimal plain decompressor and prefix one, but they posed the following question: is it true that the domain of every optimal plain decompressor contains the domain of some optimal prefix decompressor? We answer this question in the negative:

\begin{theorem}
There exist an optimal plain decompressor~$D$ with domain $S$ such that no set $T\subset S$ is the domain of an optimal prefix-free decompressor.
\end{theorem}

Note that for every decidable set $A$ that contains a fixed fraction (say, at least one third) of $n$-bit strings for every $n$, there is an optimal plain decompressor whose domain is a subset of $A$. Indeed, in this case there exists an injective mapping $p\mapsto a(p)$ such that for every string $p$ the string $a(p)$ belongs to $A$ and is two bits longer than~$p$. (Our assumption guarantees that there is enough strings of this length in~$A$.) Then let us take any optimal decompressor $U$ and replace $k$-bit descriptions by $(k+2)$-bit descriptions inside $A$: let $V(a(p))$ be equal to $U(p)$. Then $V$ is an optimal decompressor whose domain is a subset of $A$. (A more general question: which sets are the domains of an optimal plain decompressor? --- is answered in~\cite{cnss}.) 

So it is enough to show that there exists a decidable set $A$ with this property (containing at least $1/3$ of $n$-bit strings for every $n$) such that there is no optimal prefix-free decompressor whose domain is a subset of $A$. From now on we forget about plain decompressors: we need only to construct such a set $A$. This construction is provided in the next section; in the rest of this section we discuss the intuition behind it and the result itself.

The useful tool in the prefix complexity theory is provided by an observation often called \emph{Kraft--Chaitin lemma}. Consider the following ``memory allocation'' game: at each round Alice gives a natural number $n$ and Bob replies with a string of length $n$. The restriction for Alice is that the sum of $2^{-n}$ for all her numbers does not exceed $1$; the restriction for Bob is that none of his strings is a prefix of another one. Kraft--Chaitin lemma says that Bob has a computable winning strategy in this game. (See, e.g., \cite{uppsala-notes}, p.~28.)

Informally, the question posed in~\cite{cnss} asks whether this remains true if some strings (a fixed fraction for every length) are forbidden for Bob (and the allowed sum for Alice is adjusted accordingly). The answer is no: one can choose the forbidden part of every $\{0,1\}^n$ in such a way that it cripples Bob's ability to win. Technically, we need to consider a more complicated game, since complexity is defined up to a constant. We do not explain this game in details (but note that the game approach that goes back to Andrej Muchnik~\cite{muchnik-game} was an important tool for~us). Instead, we give a self-contained proof that combines game-theoretic and recursion-theoretic arguments.

Finally, one may say that the question itself is a bit artificial: one may ask instead whether for every optimal plain decompressor there is some \emph{restriction} of it (on some smaller enumerable domain) that is an optimal prefix-free decompressor. In this form, however, the answer is negative for obvious reasons: consider an optimal plain decompressor $U$ where two different strings $s$ and $t$ have unique descriptions ($U$-preimages) $p_s$ and $p_t$, and, say, $p_s$ is a prefix of $p_t$.

\section{Construction}
\label{sec:construction}

\subsection*{Statement of the Lemma}

As explained in the previous section, it is enough to prove the following lemma:

\begin{lemma}
There exists a decidable \textup{(}synonyms: computable, recursive\textup{)} set $A$ of strings with the following properties:

\textup{(1)}~For every $n$ the set $A$ contains at least $1/3$ of all strings of length $n$;

\textup{(2)}~There is no optimal prefix-free decompressor whose domain is a subset of~$A$.
\end{lemma}

\subsection*{Constructing $A$}

Describing the construction, we identify binary strings with vertices of the full binary tree: empty string is the root, string $x$ has sons $x0$ and $x1$. The set $\Omega = \{0, 1\}^{\omega}$ of infinite binary sequences is identified with $[0,1]$. For each string $x$ we define an interval $I_x\subset [0,1]$; empty string corresponds to the entire $[0,1]$; the intervals $I_{x0}$ and $I_{x1}$ are left and right halves of $I_x$. In $\Omega$ the interval $I_x$ corresponds to the subtree that consists of binary sequences that have prefix $x$, and we use the notation $I_x$ both for intervals in $[0,1]$ and in $\Omega$.

The intervals $I_x$ are called \emph{basic} intervals in the sequel. A \emph{basic subset} of $\Omega$ is a finite union of basic intervals; we may assume without loss of generality that these intervals have the same length and are disjoint, i.e., correspond to different vertices at the same level of the tree. (In $[0,1]$ we consider intervals that share an endpoint as disjoint.) If a basic set $V$ equals the union $\cup_{x\in X} I_x$ where $X \subset \{0,1\}^n$, we say that $X$ \emph{represents $V$ at level $n$}. Each basic set can be represented at all sufficiently high levels.

We construct the set $A$ layer by layer in such a way that every basic set of measure at least $1/3$ is represented by some layer of $A$: for every basic set $V$ of measure at least $1/3$ there exists $n$ such that $A\cap \{0,1\}^n$ represents $V$ at level $n$. (In a sense, this makes $A$ ``universal'': every possible restriction appears somewhere.) Moreover, every basic set $V$ (of measure at least $1/3$) should be represented by infinitely many layers that form large groups of subsequent layers: there are infinitely many $n$ such that $V$ is represented by $A$ at levels $n,n+1,\ldots, 2n$). It is easy to find a decidable set $A$ with this property (the family of all basic sets is countable and can be effectively enumerated, so we allocate infinitely many groups of layers for every basic set).

It remains to show (assuming that $A$ has these properties) that no optimal prefix-free decompressor can have a domain that is a subset of $A$.

\subsection*{Density}

Assume that $D$ is an optimal prefix-free decompressor whose domain is a subset of $A$.  The strings $x$ where $D$ is defined form a prefix-free set.
The corresponding intervals $I_x$ are disjoint; let $\mathbf{D}\subset\Omega$ be the union of these intervals.

\begin{lemma}[Density]
$\mathbf{D}$ intersects  any basic set of measure at least $1/3$.
\end{lemma}

\begin{proof}
Let $V$ be a basic set of measure at least $1/3$. According to the assumption, there are infinitely many $n$ such that $V$ is represented by $A$ at all levels $n,n+1,\ldots,2n$. If $\mathbf{D}$ does not intersect $V$, this  implies that $D$ is undefined on strings of lengths $n\ldots 2n$, which is impossible for an optimal $D$ (most of the strings of length $1.5n$ have complexity between $n$ and $2n$ for large values of $n$, so description of those lengths should exist).	
\end{proof}

\subsection*{Splitting the task}

Prefix complexity can be equivalently defined as the logarithm of maximal lower semicomputable semimeasure. We use this fact in one direction: if $x\mapsto q(x)$ is a lower semicomputable non-negative function and $\sum_x q(x)\le 1$, then $K(x)\le -\log_2 q(x) +O(1)$. We need to get a contradiction and show that every prefix-free function $D$ defined on the subset of $A$ is not an optimal prefix-free decompressor. For this purpose we construct a lower semicomputable non-negative function $x\mapsto q(x)$ and show that $C_D(x)\le -\log_2 q(x)+O(1)$ is false.

Replacing $O(1)$-notation by an explicit statement, we obtain the following claim:

\begin{quote}
for every $c=1,2,\ldots$ there exists $x$ such that $C_D(x)\ge-\log_2 q(x)+c$.
\end{quote}

We achieve this by constructing for each $c$ a (uniformly) lower semicomputable function $x\mapsto q_c(x)$ such that  $\sum_x q_c(x) \le 2^{-c}$ and $C_D(x)\ge-\log_2 q_c(x)+c$ for some $x$. Then we let $q(x)=\sum_c q_c(x)$; the sum $\sum_x q(x)$ does not exceed $1$ since the corresponding sum for $q_c$ is bounded by $2^{-c}$ and $\sum_c 2^{-c}\le 1$. (We can use other converging series instead of $\sum_c 2^{-c}$.)

\subsection*{Constructing $q_c$}

It remains to show how one can ``lower semicompute'' (=enumerate from below) some function $q_c$ with the required property while watching the enumeration of the graph of $D$. Imagine that Alice is given some ``capital'' $2^{-c}$ and is allowed to distribute this amount between different strings~$x$ (note that we distribute capital between strings that form the image of $D$); her goal is to allocate at least $2^c\cdot 2^{-C_D(x)}$ to some $x$. Of course, Alice does not know the final value of $C_D(x)$; it can decrease later (after the allocation is made). So Alice needs to guarantee that her allocation still prevails for some $x$ independently of what happens after the allocation is done.

How can Alice achieve this goal? To explain her strategy, let us introduce some terminology. The vertices (strings) in $A$ are \emph{allowed}, and the strings outside $A$ are \emph{prohibited}.  (For each level at least $1/3$ of all strings of this length are allowed.)

These notions do not depend on time (i.e., on the number of steps in the enumeration of the domain of $D$). The other notion is dynamic. Let $\bar D$ be the part of the domain of $D$ that already appeared in the enumeration process. A string $u$ is \emph{free} at that step if $\bar D \cup \{u\}$ is prefix-free. (A string that is not free cannot appear later in the domain of $D$ since this domain should remain prefix-free.) In terms of $\Omega$ this definition can be reformulated as follows: $u$ is free if $I_u$ and the set $\overline{\mathbf{D}}$ of all sequences that have prefix in $\bar D$ are disjoint.

If at some level there are no free allowed strings, this guarantees that no new strings of this length will appear in the domain of $D$.

A free string can later become non-free but not vice versa. Note also that an extension of a free string is free, so the fraction of free strings at level $n$ is a non-decreasing function of $n$ (at any moment).

Only allowed free strings can be later used as descriptions, so if at some level and at nearby levels they form a very small minority, Alice can use this fact to achieve her goal.  Let us make this statement more precise.

\subsection*{Winning case}

Assume that at all levels in some interval (say, between $l$ and $L$) the allowed strings represent the same basic set. Then the fraction of free allowed strings of length $n$ increases as $n$ increases (from $l$ to $L$).  Assume that at level 	$L$ this fraction is still less than some small $\varepsilon >0$.

What can Alice do in this case? She can allocate $2^c \cdot 2^{-L}$ to many (say, $N$; the value of $N$ will be chosen later) different strings that have no description yet (do not belong to the image of the current part of $D$).  If this turns out to be insufficient for her to win, each of these $N$ strings gets later a description of length at most $L$ (otherwise Alice still prevails on this string). These descriptions are different (moreover, none of them is a prefix of another one). Only $2^l$ descriptions may have length less than $l$, so at least $N-2^l$ of them are in  our interval (have lengths between $l$ and $L$). All these descriptions were free when Alice made her move, so at that moment the fraction of free allowed strings of length $L$ is at least
	$$
(N-2^l)/ 2^{L}
	$$
(If a free allowed string appears at an intermediate level between $l$ and $L$, this can only increase the fraction, since it can be replaced by several free allowed strings at level $L$.)

We come to a contradiction if
	$$
(N-2^l)/2^{L}\ge \varepsilon,
	$$
i.e.,
	$$
N\ge \varepsilon \cdot 2^L + 2^l
	$$
Recall that the total capital of Alice is bounded by $2^{-c}$, so the allocated amount needed to win is
	$$
(\varepsilon \cdot 2^L + 2^ l) \cdot 2^c \cdot 2^{-L}= \varepsilon \cdot 2^c + 2^c/2^{L-l}.
	$$
Therefore, Alice wins if both $\varepsilon \cdot 2^c$ and $2^c/2^{L-l}$ are bounded by $2^{-c}/2$. Both conditions are satisfied, for example, if
	$$
\varepsilon = 2^{-3c} \ \text{ and }\ L-l \ge 3c.
	$$

\subsection*{Strategy for Alice}

We arrive to the following strategy for Alice.

For a given $c$, Alice waits until an interval $[l,L]$ appears where
	
\begin{itemize}

\item[\textbullet] $L-l \ge 3c$;

\item[\textbullet] allowed strings represent the same basic set at all levels between $l$ and $L$;

\item[\textbullet] the (current) fraction of free allowed strings at level $L$ is less than $\varepsilon=2^{-3c}$.

\end{itemize}

As soon as such an interval appears, Alice allocates $2^c \cdot 2^{-L}$ to $N=\varepsilon \cdot 2^L + 2^l$ fresh strings (that have no descriptions yet).

As we have seen, this guarantees that Alice wins, i.e., that
$q_c (u) \ge 2^c \cdot 2^{-C_D(u)}$ for one of these strings.

\subsection*{Why it helps}

It remains to show that the event that Alice is waiting for will indeed happen. Assume that it is not the case. Recall that (by our construction) every basic set is represented infinitely many times by blocks of levels, and all these blocks (except for finitely many of them) are thick enough (have $L-l\ge 3c$).  Therefore, the fraction of free allowed vertices at the bottom line of each block never drops below $\varepsilon = 2^{-3c}$.

This leads to a contradiction in the following way. Fix some block (``the first block'') that is thick enough and wait until the fraction of free allowed vertices at its bottom level stabilizes.  Let $B_0$ be the basic set that is represented by the set of free allowed vertices at this level; by assumption, its measure is at least $\varepsilon$.

If the measure of $B_0$ is at least $1/3$, we get a contradiction with density lemma. So it is less than $1/3$ (and therefore $2/3$), so there exists second block below the first one where prohibited (=not allowed) elements represent $B_0$. At the bottom line of this block the fraction of free allowed strings also never drops below $\varepsilon$. Wait until it stabilizes and let $B_1$ be the basic set that corresponds to the free allowed strings at this level.  By construction $B_0$ and $B_1$ are disjoint (we considered only allowed strings while constructing $B_1$, and $B_0$ corresponds to prohibited strings).

If the measure of $B_0 \cup B_1$ is at least $1/3$, we again get a contradiction with density lemma (since $B_0\cup B_1$ and $\mathbf{D}$ are disjoint; recall that we wait for the stabilization). So we can find a third block where $B_0\cup B_1$ is prohibited, wait for the stabilization at its bottom line, construct $B_3$ etc.

Finally we get a contradiction since each block contributes at least $\varepsilon$ to the measure and at some point we exceed the threshold $1/3$.

Technical remarks:  (1) The threshold $1/3$ can be replaced by any other value not exceeding $1/2$: we need to get a contradiction before the size of the prohibited part becomes too large. In our argument we may prohibit up to $2/3$ of all strings and $1/3$ is enough for a contradiction.

(2) The construction of $B_0,B_1,\ldots$ is not effective but this is not necessary since we only prove the existence of a moment when the fraction of free allowed strings drops below~$\varepsilon$.

\section{Acknowledgments}

We thank all the participants of Kolmogorov seminar and Undergraduate Seminar at the Logic and Theory of Algorithms division of Mathematics Department, Moscow Lomonosov State University.


\begin{thebibliography}{9}
\bibitem{cnss}
Cristian S.~Calude, Andr\'e Nies, Ludwig Staiger, Frank Stephan,
Universal recursively enumerable sets of strings. In: \emph{Developments in
Language Theory, 2008}, Lecture Notes in Computer Science, 5257 (2008),
p.~170--182.

\bibitem{muchnik-game}
Andrej A. Muchnik, Ilya Mezhirov, Alexander Shen, Nikolay Vereshchagin,
\emph{Game interpretation of Kolmogorov complexity}, \texttt{arxiv:1003.4712}

\bibitem{uppsala-notes}
Alexander Shen, \emph{Algorithmic Information Theory and Kolmogorov Complexity}. Lecture notes of an introductory course at Uppsala university,
available at \texttt{www.it.uu.se/research/publications/reports/2000-034}.

\end{thebibliography}
\end{document}